\documentclass[letterpaper, 10 pt, conference]{ieeeconf}  
\usepackage{gen_settings}

\IEEEoverridecommandlockouts                              

\title{\LARGE \bf
Learning Performance Bounds for Safety-Critical Systems
}

\author{Prithvi Akella, Ugo Rosolia, and Aaron D. Ames
}

\begin{document}

\maketitle
\thispagestyle{empty}
\pagestyle{empty}

\begin{abstract}
As the complexity of control systems increases, the need for systematic methods to guarantee their efficacy grows as well.  However, direct testing of these systems is oftentimes costly, difficult, or impractical. As a result, the test and evaluation ideal would be to verify efficacy of a system simulator and use this verification result to make a statement on true system performance.  This paper formalizes that performance translation for a specific class of desired system behaviors.  In that vein, our contribution is twofold.  First, we detail a variant on existing Bayesian Optimization Algorithms that identifies minimal upper bounds to maximization problems, with some minimum probability.  Second, we use this Algorithm to $i)$ lower bound the minimum simulator robustness and $ii)$ upper bound the expected deviance between true and simulated systems.  Then, for the specific class of desired behaviors studied, we leverage these bounds to lower bound the minimum true system robustness, without directly testing the true system.  Finally, we compare a high-fidelity ROS simulator of a Segway, with a significantly noisier version of itself, and show that our probabilistic verification bounds are indeed satisfied.
\end{abstract}

\section{INTRODUCTION}
The notion of and development of controllers for \textit{safety-critical systems} has seen a tremendous rise in importance in the recent past.  Succinctly, the notion refers to the control of systems where safety is of paramount importance, \textit{e.g.} autonomous vehicles, surgical robotics, robots that have to interact with humans, \textit{etc}.  However, as is commonly the case in control development, the controller is first developed with respect to a model/simulator of the control system at hand and then deployed on the real system.  As ensuring system safety is of paramount importance however, this naturally begs the question, \textit{how do we verify true system safety}?

This question is the subject of widespread study in the Test and Evaluation community~\cite{Althoff2018, Koschi2019, Wheeler2019, tuncali2016utilizing, fainekos2012verification}.  To frame this question, desired system behaviors are oftentimes expressed as temporal logic specifications~\cite{baier2008principles,corso2020survey}.  The Test and Evaluation goal is then to iteratively develop more difficult tests of system behavior until either the system has been verified, or a failure is found - such failures are termed counterexamples.  Indeed there exist state of the art software to perform this counterexample search for system simulators~\cite{annpureddy2011s,donze2010breach,dreossi2019verifai}.  Additionally, this counterexample search is oftentimes phrased as an optimization problem - see Section III in~\cite{corso2020survey}.  This has prompted the study of how specific optimization procedures - Bayesian Optimization in particular - can provide probabilistic verification results for system simulators~\cite{ghosh2018verifying,gangopadhyay2019identification,deshmukh2017testing}.  Indeed Bayesian Optimization has also been useful for iterative control development as well~\cite{berkenkamp2016bayesian,berkenkamp2017safe,berkenkamp2015safe}.

\begin{figure}
    \centering
    \includegraphics[width = 0.48\textwidth]{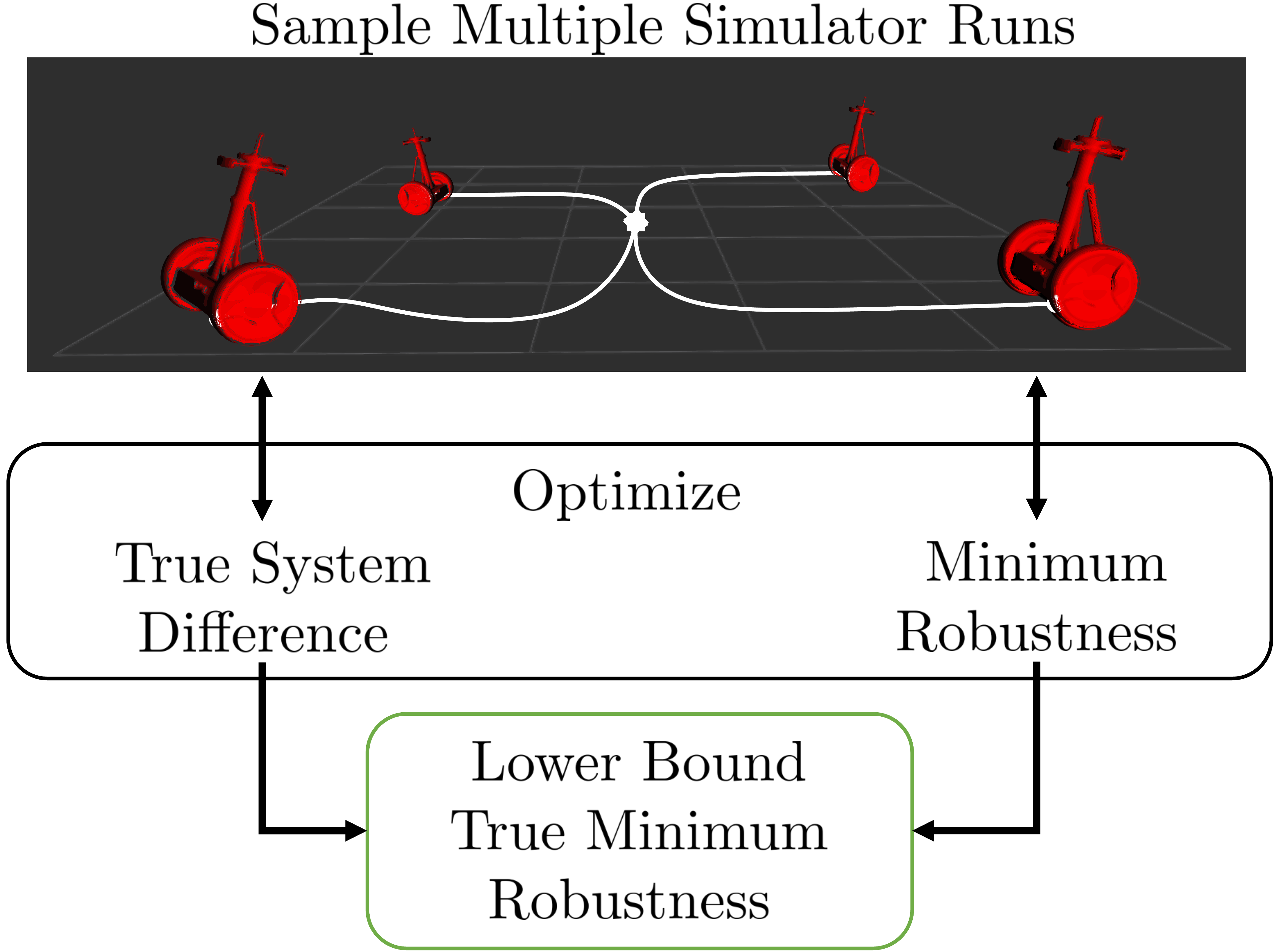}
    \gapclose
    \caption{Shown above is a flowchart for the procedure presented in the paper.  A Bayesian Optimization procedure queries a system simulator to lower bound the minimum simulator robustness and upper bound the maximum true system (semi)norm difference.  We leverage these bounds to lower bound a robustness risk measure for the true system, to some minimum probability.}
    \label{fig:title}
    \vspace{-0.25 in}
\end{figure}

However, existing work poses a few questions which we aim to address in our work.  Specifically, as posed in~\cite{corso2020survey}, attempting to directly apply the optimization based techniques prevalent in~\cite{ghosh2018verifying,gangopadhyay2019identification,deshmukh2017testing,berkenkamp2015safe,berkenkamp2017safe} to real-systems might require a prohibitively large number of true system samples to realize any counterexamples.  A Random Embeddings approach, as done in~\cite{deshmukh2017testing}, would facilitate scaling to higher dimensional optimization problems, but for any given dimension, can we minimize the number of true system evaluations required to verify or find counterexamples to true system behavior?  Additionally, can we exploit system simulators to offset the true-system evaluation cost?  Finally, can we build off existing Bayesian-based verification techniques and provide an algorithm that guarantees a prescribed tolerance to the optimal value, to further increase confidence in the verification statement or counterexample found?

\spacing
\newidea{Our Contribution:} Our contribution is twofold.  First, we detail a variant on existing, Bayesian Optimization Algorithms, that identifies minimal upper bounds to maximization problems to within a prescribed tolerance and minimum probability.  Secondly, we identify and prove that we can bound the solutions to two, simulator-based optimization problems via the prior optimization procedure.  Then, for a specific subclass of system specifications, we show these simulator-based bounds can be used to lower bound a robustness risk measure for the true system \textit{without having to directly test the true system and with fewer evaluations than direct, true-system testing}.  Finally, we demonstrate our results by lower bounding the risk measure of a highly noisy system simulator and make a comparison with respect to direct testing of the same simulator.

\spacing
\newidea{Organization:} In Section~\ref{sec:preliminaries} we outline some necessary background math. In Section~\ref{sec:statement} we formally state the problem under study.  In Section~\ref{sec:algorithm} we detail our proposed Bayesian Optimization Algorithm and prove a Theorem regarding its use.  In Section~\ref{sec:accuracy}, we employ this Algorithm to prove that we can lower bound a robustness risk measure for the true system.  Finally, Section~\ref{sec:bounding_test} shows an example use-case of Algorithm~\ref{alg:algorithm}, and Section~\ref{sec:bounding_robustness} shows an example wherein we lower bound this risk measure for a high-fidelity ROS simulator of a Segway, with its noisy counterpart.

\section{Problem Formulation}
In this section, we introduce some necessary background material.  Then, we formally state the problem under study.

\subsection{Mathematical Preliminaries}
\label{sec:preliminaries}
This section will detail some necessary background math.  We will start with some notation.

\spacing
\newidea{Notation:} $\mathbb{R}_+ = \{x \in \mathbb{R}~|~x \geq 0\}$, and $\mathbb{R}_{++} = \{x \in \mathbb{R}~|~x > 0\}$.  A multivariate function $f:\mathbb{X} \times \mathbb{Y} \to \mathbb{R}$ is partially-Lipschitz if there exists a (semi)norm $\|\cdot\|$ on $\mathbb{X}$ and a strictly positive constant $L \in \mathbb{R}_{++}$ such that $|f(x,y) - f(z,y)| \leq L \|x-z\|$.  This (semi)norm may depend on $y\in\mathbb{Y}$.  Finally, a signal $s: \mathbb{R}_+ \to \mathbb{R}^n$.  The set of all $n$-dimensional, real-valued signals $\signalspace = \{s: \mathbb{R}_+ \to \mathbb{R}^n~|~s(t) = x \in \mathbb{R}^n\}$.

\spacing
\newidea{Signal Temporal Logic:} Signal Temporal Logic (STL) is a language by which rich, time-varying system behavior can be succinctly expressed.  This language is based on atomic propositions $\phi \in \mathcal{A}$ which are boolean valued variables dependent on predicate functions $\mu: \mathbb{R}^n \to \mathbb{R}$:
\begin{equation}
    \phi(x) = \true \iff x \in \llbracket \phi \rrbracket = \{x \in \mathbb{R}^n~|~\mu(x) \sim b \}.
\end{equation}
Here, $\mathcal{A}$ is the set of all atomic propositions, $\phi(x)$ denotes the truth evaluation of the proposition $\phi$ at the state $x$, $b \in \mathbb{R}$, and $\sim~= \{\geq, \leq, <, >\}$ \cite{maler2004monitoring,baier2008principles}.  System specifications $\psi$ can be defined as follows, with $"|"$ demarcating definitions:
\begin{align}
    \label{eq:spec}
    \psi & \triangleq \phi | \neg \psi | \psi_1 \lor \psi_2 | \psi_1 \wedge \psi_2 | \psi_1 \until_{[a,b]} \psi_2,~ \psi \in \mathbb{S}.
\end{align}
Here, $\psi_1,\psi_2$ are specifications themselves, and $\psi_1 \until_{[a,b]} \psi_2$ reads as: $\psi_1$ should be true at time, $t = a$ and should continue to be true until $\psi_2$ is true, which should be true by some time, $t \leq b$ \cite{baier2008principles,maler2004monitoring}.
Finally, $\mathbb{S}$ is the set of all STL specifications.

We write $s(t) \models \psi$ when a signal $s$ satisfies a specification $\psi$ by time $t$.  Here, $\models$ is termed the satisfaction relation and is defined as follows:
\begin{align}
    & s(t) \models \phi \iff \phi(s(t)) = \true, \\
    & s(t) \models \neg \psi \iff s(t) \not \models \psi, \\
    & s(t) \models \psi_1 \lor \psi_2 \iff s(t) \models \psi_1 \lor s(t) \models \psi_2, \\
    & s(t) \models \psi_1 \wedge \psi_2 \iff s(t) \models \psi_1 \wedge s(t) \models \psi_2, \\
    & s(t) \models \psi_1 \until_{[a,b]} \psi_2 \iff \exists~t^*\leq \min\{b,t\}\suchthat \\
    & \quad \left(s(t') \models \psi_1~\forall~a \leq t' \leq t^*\right) \wedge \left(s(t^*) \models \psi_2 \right).
\end{align}

\spacing
\newidea{Gaussian Processes:}
The treatment of Gaussian Processes (GP) in this subsection stems primarily from~\cite{rasmussen2003gaussian} and will be specific to their applications in Machine Learning/Bayesian Optimization.  In this setting, GPs permit estimation and facilitate optimization of (perhaps) nonlinear, black-box functions $J: \mathbb{Z} \to \mathbb{Y}$.  Specifically, consider a set of points, $A_n = \{z_1, z_2, \dots, z_n\},~z_i \in \mathbb{Z}$ and a corresponding set of noisy evaluations of $J$ $\{y_1, y_2, \dots, y_n\}$.  As defined in~\cite{chowdhury2017kernelized}, we will assume the samples,
\begin{equation}
    \label{eq:sampling_criteria}
    y_i = J(z_i) + \xi_i,~\xi_i \sim \mathcal{N}(0,\lambda v^2),
\end{equation}
where $\lambda,v$ are GP regression parameters.  By treating the outputs $y_i$ as random variables, we can fit a GP $\pi$ to $J$ based on our choice of kernel function $k: \mathbb{Z} \times \mathbb{Z} \to \mathbb{R}_+$.
\begin{align}
    \label{eq:GP_mean}
    \mu_n(z) & = k_n(z)^T\left(K_n + \lambda I\right)^{-1}y_{1:n}, \\
    k_n(z,z') & = k(z,z') - k_n(z)^T\left(K_t + \lambda I\right)^{-1}k_n(z'), \\
    \label{eq:GP_sigma}
    \sigma_n(z) & = k_n(z,z).
\end{align}
Here, $k_n(z) = [k(z,z_1), \dots k(z,z_n)]^T$ is the covariance of $z$ with respect to the sampled data $z_i\in A_n$, $y_{1:n} = [y_1, y_2, \dots, y_n]^T$ are the noisy samples, and $(K_n)_{i,j} = k(z_i,z_j),~z_i,z_j\in A_n$ is the positive-definite Kernel Matrix. 

Finally, for any kernel $k$ there exists an associated space of functions estimate-able by said kernel, its Reproducing Kernel Hilbert Space (RKHS) $\mathcal{H}_k$.  Mathematically, for some set $Z$ and a Hilbert Space $\mathcal{H}_k$ of real-valued functions $J$ over $Z$, with a linear evaluation functional,
\begin{equation}
    L_z(J) = J(z),~\forall~J \in \mathcal{H}_k,~z \in Z,
\end{equation}
$\mathcal{H}_k$ is a RKHS if $\forall~z\in Z~\exists~M_z \geq 0$ such that, 
\begin{equation}
|L_z(J)| = |J(z)| \leq M_z \|J\|_{RKHS},~\forall~J\in\mathcal{H}_k.
\end{equation}
Here, $\|J\|_{RKHS}$ is the RKHS norm of the function $J \in \mathcal{H}_k$, defined with respect to the linear operator $L_x$.

\spacing
\newidea{Bayesian Optimization:} The brief description of Bayesian Optimization (BO) in this subsection stems primarily from~\cite{srinivas2009gaussian, chowdhury2017kernelized}.  Bayesian Optimization attempts to solve optimization problems of the following form:
\begin{equation}
    J^* = \max_{z \in \mathbb{Z}}~J(z),~\mathbb{Z}\subset \mathbb{R}^l,~l<\infty.
\end{equation}
Here, $\mathbb{Z}$ is typically a hyper-rectangle or some, compact set, for which membership is easily identifiable.  Additionally, the function $J$ is typically a non-convex, black-box function, for which gradients are not easily accessible.  The optimization procedure follows a series of steps.  First, either a Gaussian Process is provided or fit to an initial data-set $\mathbb{D}_0 = \{(z_i,y_i)\}_{i=1}^n$ with samples $y_i$ as in~\eqref{eq:sampling_criteria}.  Second, the next sample point $z_{i+1}$ is defined as the maximizer of an \textit{acquisition function} over the fitted Gaussian Process $\pi$ to the function $J$:
\begin{equation}
    \label{eq:UCB}
    z_{i+1} = \argmax_{z \in \mathbb{Z}}~\mu_i(z) + \beta_{i+1} \sigma_i(x).
\end{equation}
The Upper Confidence Bound (UCB) acquisition function is shown above, and is one, example acquisition function~\cite{srinivas2009gaussian,mockus1978application, kandasamy2018parallelised}.  Third, the procedure samples $z_{i+1}$, generates a new measurement $y_{i+1}$, adds it to the data-set, fits another Gaussian Process and repeats the procedure.

Additionally, BO procedures guarantee convergence by proving sub-linear growth in the sum-total regret $R_j = \sum_{i=1}^j r_i$ where $r_i = J^*-J(z_i)$. For certain kernels $k$ and choices of $\beta_{i+1}$, there are Bayesian Optimization procedures that guarantee sub-linear regret growth~\cite{chowdhury2017kernelized,srinivas2009gaussian}.  As we assume we have noisy samples $y_i$ of our function $J$, these regret growth bounds are written with respect to the \textit{maximum information gain} at iteration $i$, which, for our specific setting, is as follows:
\begin{gather}
    \label{eq:max_info_gain}
    \gamma_i = \max_{A \subset \mathbb{Z} \suchthat |A| = i}~I(y_A;J_A).
\end{gather}
Here, $I(y_A;J_A)$ is the \textit{mutual information gain} between $J_A = [J(z)]_{z\in A}$ and $y_A = J_A + \xi_A \sim \normal(0,\lambda v^2I)$, and quantifies the reduction in uncertainty about the objective $J$ after sampling points $z \in A$.

\subsection{Problem Statement}
\label{sec:statement}
To start, we consider our safety-critical system to be an uncertain, closed-loop control system:
\begin{equation}
    \label{eq:true}
    \dot x = f(x,u,d,w),~u(t) = U(x(t),d),
\end{equation}
where the system state $x \in\mathbb{R}^n$, the control input $u\in\mathcal{U} \subseteq \mathbb{R}^m$, the vector of known, variable phenomena $d \in \mathcal{D} \subseteq \mathbb{R}^p$, the unknown disturbance $w$ is defined via the unknown distribution $\pi_{true}$, and the controller $U$, is a feedback controller accounting for the variables $d$.  Examples of variable phenomena $d$ would be the initial condition of the system $x_0$, uncertain obstacle locations, \textit{etc}.  Additionally, we presume the controller for the true closed-loop system was built with respect to a nominal model and controller:
\begin{equation}
    \label{eq:sim}
    \dot{\hat x} = \hat f(\hat x, u, d, \hat w),~ \quad u(t) = \hat U(\hat x(t),d),
\end{equation}
with the same spaces as before for all appropriate variables, and with the addition that the simulator noise $\hat w$ is defined with respect to the unknown distribution $\pi_{nom}$.

Without loss of generality, we can assume that the vector of variable phenomena $d$ codifies the initial condition $x_0$ for either closed-loop system~\eqref{eq:true} or \eqref{eq:sim}.  As a result, the closed-loop nominal and true system trajectories, $\hat s, s \in \signalspace$, are uniquely determined by choice of variable phenomena $d$ and resulting noise sequences $W$ and $\hat W$.  More accurately,
\begin{align}
    s(t) = s_0 + \int_{l = 0}^t f\left(s(l),U\left(s(l),d\right),d,W(l)\right) dl, \\
    \hat s(t) = s_0 + \int_{l = 0}^t f\left(\hat s(l),\hat U\left(\hat s(l),d\right),d,\hat W(l)\right) dl,
\end{align}
where the vector of variable phenomena $d$ includes the initial condition $s_0$ among other phenomena $d'$, i.e., $d = [s_0, d']$.
We will represent these $d$-dependent closed-loop trajectories as
\begin{gather}
    \label{eq:signals}
    \hat \Sigma (d) = \hat s \sim \Pi_{nom}(d), ~\Sigma(d) = s \sim \Pi_{true}(d),
\end{gather}
for the unknown distributions $\Pi_{nom}(d)$, $\Pi_{true}(d)$, over $\signalspace$.  Finally, we assume the true system is to satisfy an operational, STL specification $\psi$~\eqref{eq:spec} equipped with a robustness measure $\rho$, defined as follows (inspired by~\cite{Madsen2018}):

\begin{definition}
For any signal temporal logic specification $\psi$ there exists a robustness measure $\rho$ which quantifies how robustly a given signal satisfies/does not satisfy $\psi$:
\begin{gather}
    \label{eq:robMeasure}
    \begin{split}
        \hspace{-0.1 in}\rho: \signalspace \times~\mathbb{R}_+ \to \mathbb{R} \suchthat      \rho(s,t) \geq 0 \iff s(t) \models \psi.
    \end{split}
\end{gather}
\end{definition}
As the true system signal $\Sigma(d)$ as per equation~\eqref{eq:signals} is a random variable, we assume the verification goal is to minimize a risk measure for the system robustness at some evaluation time $T$:
\begin{equation}
    \label{eq:robust_real}
     \rho^* = \min_{d \in \mathcal{D}}~\expect\left[\rho\left(\Sigma(d),T\right)\right] - r \sqrt{\var\left(\rho\left(\Sigma(d),T\right)\right)}.
\end{equation}
Here $r>0$, and the moments are calculated with respect to the distribution $\Pi_{true}(d)$.

\begin{remark}
The reason we deal with a risk measure in equation~\eqref{eq:robust_real} as opposed to optimizing for the minimum expected robustness, is that a system may realize a positive minimum expected robustness yet still almost always fail to satisfy a specification~\cite{artzner1999coherent,ahmadi2020constrained,safaoui2020control}.  For example, consider a system where $\rho(\Sigma(d),T)<0$ for the vast majority of runs, but for a few rare cases, $\rho(\Sigma(d),T) >>0$, such that the expected robustness is positive, but the system almost always fails to satisfy its specification.  Optimizing for a risk measure accounts for this case, by discounting the objective with the variance of the robustness measure at any $d$.
\end{remark}

\spacing
\newidea{Our Approach:} While we could directly attempt to solve for $\rho^*$ from~\eqref{eq:robust_real} via Bayesian Optimization, such a procedure would require multiple runs of the true system, and \textit{ideally we would like to find a large lower bound to $\rho^*$ while minimizing the number of true system evaluations.  Hence, we will instead solve optimization problems over functions of the nominal and true systems, $\hat \Sigma(d)$ and $\Sigma(d)$, variable phenomena $d$ and evaluation time $T$:}
\begin{gather}
    \hat \rho(d) = \expect_{\Pi_{nom}(d)}\left[\rho\left(\hat \Sigma(d),T\right)\right], \label{eq:nom_expect_rob}\\
    e(d) = \expect_{\Pi_{nom}(d),\Pi_{true}(d)}\left[\left\|\Sigma(d) - \hat \Sigma(d) \right\|\right], \label{eq:sim_gap} \\
    \label{eq:sim_opt}
    \hat \rho^* = \min_{d \in \mathcal{D}}~\hat \rho(d),~e^* = \max_{d \in \mathcal{D}}~e(d).
\end{gather}
Here, $\|\cdot\|$ is a (semi)norm over $\signalspace$.  Then, our procedure $i)$ bounds the solutions to optimization problems~\eqref{eq:sim_opt} via a Bayesian Optimization procedure we develop, and $ii)$ leverages these bounds to construct a lower bound for $\rho^*$ to minimum probability $1-\delta$, for some $\delta \in (0,1)$.  To note, $\hat \rho^*$ only defines the minimum expected nominal robustness and is not a risk measure like $\rho^*$, as it does not account for the variance of the robustness measure at any given $d$.  However, Popoviciu's inequality~\cite{popoviciu1935equations} permits bounding of the variance term in equation~\eqref{eq:robust_real}, thereby allowing us to solve a simpler, expected value problem for $\hat \rho^*$ to lower bound $\rho^*$.  This leads to our problem statement.

\begin{problem}
Let $\pi_1$ and $\pi_2$ be two Gaussian Processes fit to datasets of sample-measurement pairs of the functions $\hat \rho$ in equation~\eqref{eq:nom_expect_rob} and $e$ in equation~\eqref{eq:sim_gap}, respectively.  For some $\delta \in (0,1)$, determine a sufficiently large lower bound $p$ such that $\prob_{\pi_1,\pi_2}[\rho^* \geq p] \geq 1-\delta$, with $\rho^*$ as in equation~\eqref{eq:robust_real}.
\end{problem}

\section{A Bayesian-based Bound Finder}
\label{sec:algorithm}
In this section, we will detail our variant on existing GP-UCB Bayesian Optimization Procedures, that identifies minimal upper bounds $\epsilon$ to solutions to optimization problems of the following form that meet an Assumption to-be-written:
\begin{equation}
    \label{eq:base_bayesian}
    J^* = \max_{z \in \mathbb{Z}}~J(z),~\mathbb{Z} \subset \mathbb{R}^l,~l<\infty.
\end{equation}
We construct such an algorithm, as we will require accurate estimates of the minimum simulator robustness $\hat \rho^*$ and true system (semi)norm difference $e^*$ in order to lower bound the true system minimum robustness $\rho^*$.  To develop our algorithm, we use the IGP-UCB Algorithm in~\cite{chowdhury2017kernelized} referenced in prior works in the controls literature~\cite{ghosh2018verifying,gangopadhyay2019identification,berkenkamp2016bayesian,berkenkamp2017safe}.

Before stating our main result for this section, we will briefly describe our Algorithm~\ref{alg:algorithm}.  To start, we require positive constants $\delta \in (0,1)$, $B,R,\alpha,c \in \mathbb{R}_{++}$, and an initial data-set $\mathbb{D}_0 = \{(z,y)\}$, where the sample $y$ satisfies equation~\eqref{eq:sampling_criteria}. Then, Algorithm~\ref{alg:algorithm} first defines in Line 2, a scale factor
\begin{equation}
    \beta_i = B + R \sqrt{2 \ln{\frac{\sqrt{\det\left((1+\frac{2}{i})I + K_i\right)}}{\delta}}}, \label{eq:beta}
\end{equation}
and, in Line 3, identifies the maximizer of the UCB acquisition function $z_i$ (as per equation~\eqref{eq:UCB}) with respect to this $\beta_i$ and the fitted Gaussian Process $\pi$ to $J$ at iteration $i$.  In Line 4, the algorithm collects a noisy measurement $y_i$ of $J(z_i)$, and the sample pair $(z_i,y_i)$ is added to the data-set generating $\mathbb{D}_i$.  Line 5 defines the simple regret bound
\begin{equation}
    \label{eq:F}
    F_i = 2 \beta_i \sigma_{i-1}(z_i).
\end{equation}
Here, $\sigma_{i-1}$ is the variance of the fitted Gaussian Process to the data-set $\mathbb{D}_{i-1}$. Lines 6-9 check whether $F_i \leq \alpha$, the desired tolerance, and if so, the algorithm outputs $\epsilon = y_i + \alpha + c$ and terminates.  Otherwise, in Line 10, the algorithm updates the fitted Gaussian Process $\pi$ with respect to $\mathbb{D}_i$. 

We will now state an assumption used throughout this section and in other works utilizing Bayesian Optimization~\cite{ghosh2018verifying,berkenkamp2016bayesian}.  Then, we will move to the first main result of our paper.
\begin{assumption}
\label{assump:bayesian}
For the optimization problem~\eqref{eq:base_bayesian}, the feasible region $\mathbb{Z}$ is compact and convex.  Additionally, for some $B,R\in\mathbb{R}_{++}$, and kernel $k$, the objective function $J$ has $\|J\|_{RKHS} \leq B$, and the samples $y_i$ of $J(z_i)$, as per equation~\eqref{eq:sampling_criteria}, are corrupted by $R$-sub Gaussian Noise $\forall~i$.
\end{assumption}

\noindent This leads to the first key result of the paper wherein we show that the proposed algorithm is guaranteed to terminate and, at termination, identify a minimal upper bound $\epsilon$ to the function maximizer $J^*$ with minimum probability defined by the following function:
\begin{equation}
    \label{eq:probfac}
    \Delta(c,\delta,R) = \left(1-\frac{R}{c\sqrt{2\pi}}\exp\left(-\frac{c^2}{2R^2}\right)\right)(1-\delta).
\end{equation}

\begin{theorem}
\label{thm:algorithm}
Let Assumption~\ref{assump:bayesian} hold, let $\delta \in (0,1]$, and let $\alpha,c \in \mathbb{R}_{++}$.  At termination $i^*$, Algorithm~\ref{alg:algorithm} outputs $\epsilon = y_{i^*} + c + \alpha$ such that $\prob_{\pi}[J^* \leq \epsilon] \geq \Delta(c,\delta,R)$ with $J^*$ as in~\eqref{eq:base_bayesian}, $\Delta$ as in~\eqref{eq:probfac}, $\pi$ as in Line 10, and $y_{i^*}$ as in Line 4.
\end{theorem}

\begin{algorithm}[t]
\caption{Minimal Upper Bound Determination}\label{alg:algorithm}
\begin{algorithmic}[1]
\Require $\delta \in (0,1]$, $B,R \in \mathbb{R}_{++}$, an initial data-set $\mathbb{D}_0 = \{(z,y)~|~z \in \mathbb{Z},~y$ as per equation~\eqref{eq:sampling_criteria}$\}$, tolerance $\alpha \in \mathbb{R}_{++}$ and noise bound $c \in \mathbb{R}_{++}$

\hspace{-0.56 in} \noindent \textbf{Returns:} A Gaussian Process $\pi$, and an upper bound $\epsilon$ such that $\prob_{\pi}[J^* \leq \epsilon] \geq \Delta(c,\delta,R)$, with $\Delta$ as in~\eqref{eq:probfac}.
\Initialize $i=1$, $\eta_i = \frac{2}{i}$, Gaussian Process with mean $\mu_0$ and covariance $\sigma_0$ from the data-set, $\mathbb{D}_0$ as per equations~\eqref{eq:GP_mean} and \eqref{eq:GP_sigma}.
\While{True}
\State $\beta_i \gets B + R \sqrt{2 \ln{\frac{\sqrt{\det\left((1+\eta_i)I + K_i\right)}}{\delta}}}$
\State $z_i \gets \argmax_{z \in \mathbb{Z}}~\mu_{i-1}(z) + \beta_i\sigma_{i-1}(z)$
\State $\mathbb{D}_{i} \gets \mathbb{D}_{i-1} \cup (z_i, y_i$ as per equation~\eqref{eq:sampling_criteria})
\State $F_i \gets 2 \beta_i \sigma_{i-1}(z_i)$
\If{$F_i \leq \alpha$}
    \State $\epsilon = y_i + \alpha + c$
    \State \textbf{return} $\epsilon$ 
\EndIf
\State Update the Gaussian Process $\pi$ with mean $\mu_i$ and variance $\sigma_i$ as per~\eqref{eq:GP_mean} and \eqref{eq:GP_sigma} with respect to $\mathbb{D}_i$
\State $i \gets i+1$
\EndWhile
\end{algorithmic}
\end{algorithm}

Proving Theorem~\ref{thm:algorithm} requires two propositions and two lemmas.  The first proposition bounds both the variance of the objective $J$ with respect to the fitted Gaussian Process $\pi$ and the growth of the scale factor $\beta_i$.

\begin{proposition}[Theorem~2 in~\cite{chowdhury2017kernelized}]
\label{prop:gpr_bound}
Let $\beta_i$ be as in~\eqref{eq:beta}, $\gamma_j$ as in~\eqref{eq:max_info_gain}, $\delta \in (0,1]$, and let Assumption~\ref{assump:bayesian} hold.  With probability $\geq 1-\delta$, $|\mu_{i-1}(z) - J(z)| \leq \beta_i\sigma_{i-1}(z)~\forall~i$, and
\begin{gather}
    \beta_i \leq B + R \sqrt{2 \left(\gamma_j + 1 + \ln{\frac{1}{\delta} }\right)},~\forall~ i = 1,2,\dots,j.
\end{gather}
\end{proposition}
In Proposition~\ref{prop:gpr_bound}, $\mu_{i-1}$ and $\sigma_{i-1}$ are the mean/variance of the fitted Gaussian Process $\pi$ to the objective $J$ based on the data-set $\mathbb{D}_{i-1}$, and $\gamma_j$ is as in equation~\eqref{eq:max_info_gain}.  For context, both inequalities in Proposition~\ref{prop:gpr_bound} were derived in the proof of Theorem~2 in~\cite{chowdhury2017kernelized}. The second proposition bounds the growth rate of $\gamma_j$, defined in equation~\eqref{eq:max_info_gain}

\spacing
\begin{proposition}[Theorem~5 in~\cite{srinivas2009gaussian}]
\label{prop:bound_information}
Let Assumption~\ref{assump:bayesian} hold. There exists a kernel $k$ such that the growth in the maximum information gain $\gamma_j$ satisfies the following inequality:
\begin{equation}
    \gamma_j \leq O(j^p \log(j)),~p < 0.5.
\end{equation}
\end{proposition}
As before, the growth bound in Proposition~\ref{prop:bound_information} stems directly from Theorem 5 in~\cite{srinivas2009gaussian}, which provides the growth bound for the information gain $\gamma_j$ for common kernels.  Then, our first lemma proves that the simple regret $r_i$ is upper bounded by the simple regret bound $F_i$ with probability $\geq 1-\delta$.  This probability is over the Gaussian Process $\pi$ defined in Line 10 of Algorithm~\ref{alg:algorithm}.
\spacing
\begin{lemma}
\label{lem:bound_simple_regret}
Let Assumption~\ref{assump:bayesian} hold, and let $F_i$ be as in~\eqref{eq:F}. The simple regret $r_i$ satisfies the following inequality with respect to the Gaussian Process $\pi$ (Line 10):
\begin{equation}
    \prob_{\pi}[r_i \leq F_i] \geq 1-\delta.
\end{equation}
\end{lemma}
\begin{proof}
By definition of the simple regret $r_i$, the optimal sample $z_i$ (Line 3), the simple regret bound $F_i$, and the first inequality in Proposition~\ref{prop:gpr_bound}, we have the following:
\begin{align}
    r_i & = J^* - J(z_i), \\
    & \leq \beta_i \sigma_{i-1}(z_i) + \mu_{i-1}(z_i) - J(z_i), \withprob \geq 1-\delta \\
    & \leq 2 \beta_i \sigma_{i-1}(z_i) = F_i, \withprob \geq 1-\delta. 
\end{align}
\end{proof}

Our next lemma bounds the growth of $F_i$.
\spacing
\begin{lemma}
\label{lem:Fbound}
Let Assumption~\ref{assump:bayesian} hold and let $\delta \in (0,1]$. Then, 
\begin{equation}
    \sum_{i=1}^j F_i \leq O\left(\sqrt{j}\left(B\sqrt{\gamma_j} + R \sqrt{\gamma_j\left(\gamma_j + \ln{\frac{1}{\delta}}\right)}\right)\right)
\end{equation}
with probability $\geq 1-\delta$, with respect to the Gaussian Process $\pi$ (Line 10),  and with $F_i$ as in~\eqref{eq:F}.
\end{lemma}
\begin{proof}
From the definition of the simple regret bound $F_i$ and the second inequality in Proposition~\ref{prop:gpr_bound}, we have
\begin{equation}
    \sum_{i=1}^j F_i \leq 2\left(B + R \sqrt{2 \left(\gamma_j + 1 + \ln{\frac{1}{\delta} }\right)}\right) \sum_{i=1}^j \sigma_{i-1}(z_i),
\end{equation}
with probability $\geq 1-\delta$.
The result then stems via Lemma 4 in~\cite{chowdhury2017kernelizedarxiv}, which states that $\sum_{i=1}^j \sigma_{i-1}(z_i) \leq O\left(\sqrt{j\gamma_j}\right)$.
\end{proof}

With these results, we can prove Theorem~\ref{thm:algorithm}.

\begin{proof}
The proof requires two parts.  First we need to prove that Algorithm~\ref{alg:algorithm} is guaranteed to terminate, and second, that the upper bound $\epsilon$ satisfies the desired inequality.  For the first part, we need to show that $\exists~i^* < \infty$ such that $F_{i^*} \leq \alpha$.  The proof of this follows a contradiction.  Assume instead that $\nexists~i^* < \infty$ such that $F_{i^*} \leq \alpha$, \textit{i.e.}, $F_i > \alpha~,\forall~i \geq 0$.
Then consider the running average of $F_i$ and Lemma~\ref{lem:Fbound}:
\begin{align}
    \alpha & < \lim_{j \to \infty}~\frac{1}{j}\sum_{i=1}^j F_i, \\
    & \leq \lim_{j \to \infty} O\left(\frac{\sqrt{j}\left(B\sqrt{\gamma_j} + R\sqrt{\gamma_j\left(\gamma_j + \ln{\frac{1}{\delta}}\right)}\right)}{j}\right).
\end{align}
Now, pick a kernel that satisfies the inequality in Proposition~\ref{prop:bound_information}, which is guaranteed to exist.  Then,
\begin{equation}
    \alpha < \lim_{j \to \infty} O\left(\frac{j^z \log(j)}{j}\right) = 0,~\mathrm{as}~z < 1,
\end{equation}
which is a contradiction, as $\alpha \in \mathbb{R}_{++}$.  This proves termination at some $i^* < \infty$.  It remains to identify an upper bound $\epsilon$ that satisfies the required inequality in Theorem~\ref{thm:algorithm}.

By Lemma~\ref{lem:bound_simple_regret}, the definition of simple regret $r_i$, and our noisy samples $y_i$, we have with probability $\geq 1-\delta$,
\begin{align}
    \label{eq:alpha_bound}
    \alpha \geq F_{i^*} \geq r_{i^*} = J^* - J(z_{i^*}) = J^* - y_{i^*} - \xi_{i^*},
\end{align}
for the unknown noise $\xi_{i^*}$.  As $\xi_{i^*}$ is assumed to be $R$-sub Gaussian however, then via Mill's Inequality we have that 
\begin{equation}
 \prob[\xi_{i^*} \leq c] \geq 1-\frac{R}{c\sqrt{2\pi}}e^{-\frac{c^2}{2R^2}}
\end{equation}
and as a result,
\begin{equation}
    \prob_{\pi}\left[f^* \leq \alpha + y_{i^*} + c\right]\geq \Delta(c,\delta,R).
\end{equation}
Defining $\epsilon = \alpha + y_{i^*} + c$ completes the proof.
\end{proof}

\section{Learning Performance Bounds}
\label{sec:accuracy}
In this section, we will use Algorithm~\ref{alg:algorithm} and Theorem~\ref{thm:algorithm} to bound the solutions to optimization problems~\eqref{eq:sim_opt} and subsequently lower bound the true system robustness risk measure $\rho^*$ \textit{without directly testing the true system} - this is our main result.  To do so, we require one assumption for the robustness measure $\rho$ for our system's STL specification $\psi$.
\begin{assumption}
\label{assump:lipschitz_spec}
The robustness measure $\rho$ associated with the STL specification $\psi$ is $(L,\|\cdot\|)$-partially Lipschitz, and maps signals to a bounded region on the real line, \textit{i.e.} $\rho: \signalspace \times \mathbb{R}_{+} \to [-m,M]$, where $m,M \in \mathbb{R}_{++}$.
\end{assumption}
In Assumption~\ref{assump:lipschitz_spec} above, we assume the (semi)norm with respect to which $\rho$ is partially Lipschitz is the same (semi)norm over which $e$ is defined in equation~\eqref{eq:sim_gap}.  Also, bounding robustness measures is not too restrictive, as any bounded robustness measure still satisfies their definition as per equation~\eqref{eq:robMeasure}.  Then we have our second result.
\begin{theorem}
\label{thm:performance_bound}
Let the specification $\psi$ satisfy Assumption~\ref{assump:lipschitz_spec}, and let the optimization problems~\eqref{eq:sim_opt} satisfy Assumption~\ref{assump:bayesian} with sub-Gaussian noise bounds $R_1,R_2$ respectively.  There exist $\Tilde \rho, \Tilde e, \ell \in \mathbb{R}$, $\delta_1,\delta_2\in (0,1]$, and $c_1,c_2 \in \mathbb{R}_{++}$ such that
\begin{gather}
    \prob_{\pi_1}[\hat \rho^* \geq \Tilde \rho] \geq \Delta(c_1,\delta_1,R_1),\\
    \prob_{\pi_2}[e^* \leq \Tilde e] \geq \Delta(c_2,\delta_2,R_2), \\
    \prob_{\pi_1,\pi_2}\left[\rho^* \geq \ell \right] \geq \Delta(c_1,\delta_1,R_1)\Delta(c_2,\delta_2,R_2).
\end{gather}
Here, $\hat \rho^*$ and $e^*$ are from equation~\eqref{eq:sim_opt}, $\rho^*$ is from equation~\eqref{eq:robust_real}, and $\Delta$ is from equation~\eqref{eq:probfac}.
\end{theorem}
To clarify, the probability in Theorem~\ref{thm:performance_bound} is over the Gaussian Processes, $\pi_1$ and $\pi_2$, fit to datasets generated from the functions $\hat \rho$ and $e$, respectively.  More specifically then, Theorem~\ref{thm:performance_bound} states that if these functions satisfy Assumption~\ref{assump:bayesian}, then we can bound the solutions to optimization problems~\eqref{eq:sim_opt} and construct a lower bound on the risk measure $\rho^*$, to some minimum probability.  Proving Theorem~\ref{thm:performance_bound} requires two Lemmas, proving existence of $\Tilde \rho$ and $\Tilde e$.
\begin{lemma}
\label{lem:existence_lowerbound}
Let $\hat \rho^*$ be as defined in equation~\eqref{eq:sim_opt} with its optimization problem satisfying Assumption~\ref{assump:bayesian} with sub-Gaussian noise bound $R$. For any $\delta \in (0,1]$ and $\alpha,c \in \mathbb{R}_{++}$, Algorithm~\ref{alg:algorithm} will output $\Tilde \rho$ at termination, such that $\prob_{\pi}[\hat \rho^* \geq \Tilde \rho] \geq \Delta(c,\delta,R)$, with $\Delta$ as in equation~\eqref{eq:probfac}.
\end{lemma}
\begin{proof}
This is a direct application of Theorem~\ref{thm:algorithm}.  As we know that the optimization problem corresponding to $\hat \rho^*$ satisfies Assumption~\ref{assump:bayesian}, so too does its maximization equivalent also satisfy Assumption~\ref{assump:bayesian}, \textit{i.e.},
\begin{equation}
    \hat \rho^* = \min_{d \in \mathcal{D}}~\hat \rho(d) = - \max_{d \in \mathcal{D}}~-\hat \rho(d).    
\end{equation}
Then, choose a $\delta \in (0,1]$ and $\alpha,c \in \mathbb{R}_{++}$.  For these constants, Theorem~\ref{thm:algorithm} guarantees existence of an $\epsilon$ such that
\begin{equation}
    \prob_{\pi}[-\hat \rho^* \leq \epsilon] \geq \Delta(c,\delta,R).
\end{equation}
Defining $\Tilde \rho = -\epsilon$ concludes the proof.
\end{proof}

\spacing
\begin{lemma}
\label{lem:existence_upperbound}
Let $e^*$ be as defined in equation~\eqref{eq:sim_opt} with its optimization problem satisfying Assumption~\ref{assump:bayesian} with sub-Gaussian noise bound $R$.  For any $\delta \in (0,1]$ and $\alpha,c \in \mathbb{R}_{++}$ Algorithm~\ref{alg:algorithm} will output $\Tilde e$ at termination, such that $\prob_{\pi}[e^* \leq \Tilde e] \geq \Delta(c,\delta,R)$, with $\Delta$ as in equation~\eqref{eq:probfac}.
\end{lemma}
\begin{proof}
This is a direct consequence of Theorem~\ref{thm:algorithm}, as the optimization problem for $e^*$ satisfies Assumption~\ref{assump:bayesian}.
\end{proof}

Then, the proof of Theorem~\ref{thm:performance_bound} is as follows:

\begin{proof}
(Of Theorem~\ref{thm:performance_bound}) To start, via Assumption~\ref{assump:lipschitz_spec} and Popoviciu's inequality~\cite{popoviciu1935equations}, we have that:
\begin{align}
    \rho^* & \geq\min_{d \in \mathcal{D}}~ \expect_{\Pi_{true}(d)}[\rho(\Sigma(d),T)] - \frac{r(M+m)}{2}, \\
    & = r^* - \frac{r(M+m)}{2}. \label{eq:pop_bound}
\end{align}
Then,
\begin{align}
    & r^* - \hat \rho^* \\
    & = \min_{d \in \mathcal{D}}~\expect_{\Pi_{true}(d)}\left[\rho\left(\Sigma(d),T\right)\right] - \min_{d \in \mathcal{D}}~\hat \rho(d), \\
    & \leq \max_{d \in \mathcal{D}}~ \expect_{\Pi_{true}(d)}\left[\rho\left(\Sigma(d),T\right)\right] - \hat \rho(d), \\
    & \leq \max_{d \in \mathcal{D}}~\expect_{\Pi_{nom}(d), \Pi_{true}(d)}\left[\rho\Big(\Sigma(d),T\Big) - \rho\left(\hat \Sigma(d),T\right) \right]. \\[-0.35 in]
\end{align}
To simplify notation in the remainder of the proof, we will abbreviate the expectation with respect to both distributions with an expectation with respect to a joint distribution $\Pi$,  \textit{i.e.} $\expect_{\Pi_{nom}(d), \Pi_{true}(d)}[\cdot] = \expect_{\Pi(d)}[\cdot]$.
This simplifies presentation of a similar inequality for $\hat \rho^* - r^*$ as well; specifically,
\begin{equation}
    \hat \rho^* - r^* \leq \max_{d \in \mathcal{D}}~\expect_{\Pi(d)}\left[\rho\left(\hat \Sigma(d),T\right) - \rho\Big(\Sigma(d),T\Big)\right].
\end{equation}
Without loss of generality, we can assume one of the differences is positive, \textit{i.e.} either $r^* - \hat \rho^* \geq 0$ or $\hat \rho^* - r^* \geq 0$. Taking the absolute value of both sides of the inequality for which the aforementioned difference is positive nets the following result, via partial Lipschitz continuity of $\rho$:
\begin{align}
    |\hat \rho^* - r^*| & \leq \left|\max_{d \in \mathcal{D}}~\expect_{\Pi(d)}\left[\rho\Big(\Sigma(d),T\Big) - \rho\left(\hat \Sigma(d),T\right) \right] \right|, \\
    & \leq L \max_{d \in \mathcal{D}}~\expect_{\Pi(d)}\left[\left\|\Sigma(d) - \hat \Sigma(d) \right\| \right].
\end{align}
Then, pick $\delta_1, \delta_2 \in (0,1]$ and $\alpha_1,\alpha_2,c_1,c_2 \in \mathbb{R}_{++}$.  Then with $\delta_2$ and $\alpha_2$, we know $\exists~\Tilde e$ via Lemma~\ref{lem:existence_upperbound} such that
\begin{equation}
    \prob_{\pi_2}\left[|\hat \rho^* - r^*| \leq L\Tilde e\right] \geq \Delta(c_2,\delta_2,R_2)
\end{equation}
To be specific with the Gaussian Processes used throughout the proof, we note that $\pi_2$ defines a Gaussian Process fit to a dataset of noisy samples of the difference function $e$, as in~\eqref{eq:sim_gap}.  Then, with $\delta_1$ and $\alpha_1$ and by Lemma~\ref{lem:existence_lowerbound}:
\begin{equation}
    \prob_{\pi_1,\pi_2}\left[r^* \geq \Tilde \rho - L\Tilde e\right] \geq \Delta(c_1,\delta_1,R_1)\Delta(c_2,\delta_2,R_2).
\end{equation}
Here, $\pi_1$ defines a Gaussian Process fit to a dataset of noisy samples of the expected nominal robustness function $\hat \rho$, as in~\eqref{eq:nom_expect_rob}.  Then, the result follows from equation~\eqref{eq:pop_bound} and defining $\ell = \Tilde \rho - \Tilde \epsilon - \frac{r(M+m)}{2}$:
\begin{equation}
    \prob_{\pi_1,\pi_2}\left[\rho^* \geq \ell\right] \geq \Delta(c_1,\delta_1,R_1)\Delta(c_2,\delta_2,R_2),
\end{equation}
which concludes the proof.
\end{proof}


\section{Results}
In this section we show how Algorithm~\ref{alg:algorithm} can be used to produce minimal upper bounds $\epsilon$ that satisfy the inequalities in Theorem~\ref{thm:algorithm} for a specific test function.  Additionally, we lower bound the risk measure $\rho^*$ of a high-fidelity ROS simulator of a Segway system satisfying a simple safety specification, and show that our process minimizes the number of true system tests required to realize this lower bound.

\subsection{Identifying Minimal Upper Bounds}
\label{sec:bounding_test}
To show efficacy of Algorithm~\ref{alg:algorithm}, we will use it to identify minimal upper bounds to a two-dimensional test function with noisy samples of the same function, defined below:
\begin{equation}
\label{eq:test_function}
J^* = \max_{z \in [0,5]^2 \in \mathbb{R}^2}~J(z) = \frac{\sin(z_1)\cos(z_2)}{2}.
\end{equation}
For Algorithm~\ref{alg:algorithm}, our samples $y_i = J(z_i) + \xi \sim \normal(0,0.001^2)$ as per equation~\eqref{eq:sampling_criteria}.  Then, we initialized Algorithm~\ref{alg:algorithm} fifty times with the presumed upper bound on the RKHS norm of $J$, $B = 0.25$, the sub-Gaussian noise bound $R = 0.005$, the probability requirement $\delta = 0.05$, and the required tolerances $\alpha = 0.015$ and $c = 0.01$.  We also used a Mat\'ern kernel, with $l=1,~\nu=10$, as this kernel ensures that our order growth rate for the maximum information gain $\gamma_j$ (as per~\eqref{eq:max_info_gain}) satisfies the required inequality in Proposition~\ref{prop:bound_information}.  Additionally, the Mat\'ern kernel is universal for any choice of parameters~\cite{srinivas2009gaussian}, ensuring that our function $J$ lies in the RKHS of our kernel $k$ - $J \in \mathcal{H}_k$. As can be seen in Figure~\ref{fig:termination}, in each of the fifty, independent cases wherein we ran Algorithm~\ref{alg:algorithm} to produce an upper bound $\epsilon$ to $J^*$ as in~\eqref{eq:test_function}, the simple regret bound $F_i$ decays to below the required tolerance $\alpha$.  Additionally, we know \textit{apriori} that solutions to optimization problem~\eqref{eq:test_function} $J^* = 0.5$.  As such, we expect all produced upper bounds $\epsilon \geq 0.5$.  Indeed, all produced upper bounds $\epsilon \in (0.5,0.53]$ with the distribution of bounds arising from the noisy sampling procedure.

\begin{figure}[t]
    \centering
    \includegraphics[width = 0.48\textwidth]{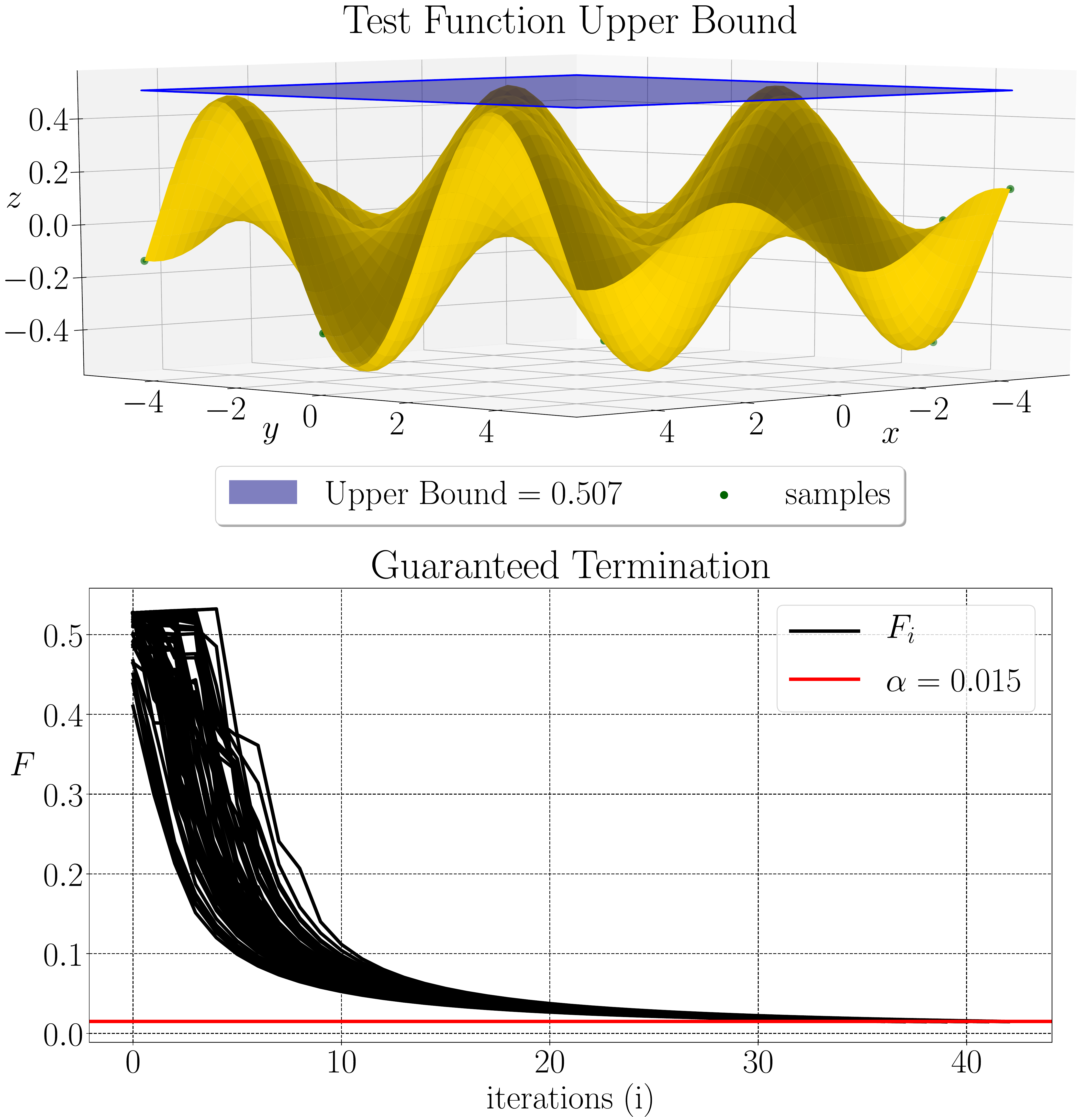}  \\
    \gapclose
    \caption{(Top) An example of the smallest upper bound, shown in blue, for the test function, $f$~\eqref{eq:test_function}, with the samples taken by the Bayesian Optimization algorithm shown in green. (Bottom) The trajectories of $F_i$~\eqref{eq:F} over fifty, independent runs of the same optimization problem, with the required tolerance, $\alpha = 0.015$, shown in red.}
    \label{fig:termination}
    \vspace{-0.5 cm}
\end{figure}

\subsection{A Safety-Critical Example}
\label{sec:bounding_robustness}
In this subsection, we leverage the results of Theorem~\ref{thm:performance_bound} to lower bound the risk measure for a Segway attempting to satisfy a safety specification.  Our nominal system is a high-fidelity ROS-based simulator of a Segway.  Our true system is the same simulator wherein we perturb the Segway's planar initial position $(x_0,y_0)$, its initial heading angle $\omega_0$, and its initial pendulum angle $\phi_0$.  Figure~\ref{fig:setup} shows our setup where the goal is to ensure the Segway's pendulum angle never deviates too far from the vertical, \textit{i.e.},
\begin{gather}
    s(t) = [x, y, \omega, \dot x, \dot y, \phi, \dot \phi]^T \in \mathbb{R}^7, \\
    \psi = \neg\left(\true \until_{[0,\infty)} \neg \left( |\phi| \leq 0.95 \mathrm{~rad}\right) \right), \\
    \rho(s,t) = 0.95 - \max_{0 \leq t' \leq t} |s_{\phi}(t')|.
\end{gather}

\noindent It is easy to verify that $\rho$ is $(1,\|\cdot\|_t)$-partially Lipschitz:
\begin{equation}
    |\rho(s,t) - \rho(z,t)| \leq \max_{0 \leq t' \leq t} |s_{\phi(t')} - z_{\phi(t')}| = \|s-z\|_t.
\end{equation}
Then, we note that we can arbitrarily bound the output of our robustness measure $\rho$ to lie within the region $[-m,M]$ where $m=0.05$ and $M=0.75$.  This still ensures our robustness measure meets its required definition in equation~\eqref{eq:robMeasure}, while maintaining the $(1,\|\cdot\|_t)$-partial Lipschitz property.

We will choose to evaluate our true system's minimum robustness risk measure $\rho^*$ at an evaluation time $T = 15$.  During this time, we will ask the Segway to navigate to a predefined goal $x_g = [2.5,2.5]$ - the center cell in the $5\times 5$ grid mention in Figure~\ref{fig:setup}.  Additionally, we expect the Segway to achieve this from anywhere in the feasible space.  Specifically, our set of variable phenomena $d$ is the planar initial condition of the Segway, \textit{i.e.} $d = s(0) = \hat s(0)$ and $d \in \mathcal{D} = [0,5]^2 \subset \mathbb{R}^2$.  Then our optimization problems~\eqref{eq:robust_real},~\eqref{eq:sim_opt} are as follows, with $r = 0.2$:
\begin{subequations}
\begin{align}
    \hspace{-0.2 in}\rho^* & = \min_{d \in \mathcal{D}}~\expect\left[\rho(\Sigma(d),15)\right] - \frac{\sqrt{\var\left(\rho(\Sigma(d),15)\right)}}{5}, \label{eq:example_true_robustness}\\
    \hspace{-0.2 in}\hat \rho^* & = \min_{d \in \mathcal{D}}~\expect_{\Pi_{nom}(d)}\left[\rho\left(\hat \Sigma(d),15\right)\right], \label{eq:example_sim_robustness}\\
    \hspace{-0.2 in}e^* & = \max_{d \in \mathcal{D}}~\expect_{\Pi(d)}\left[\left\|\Sigma(d) - \hat \Sigma(d) \right\|_{15}\right]. \label{eq:example_accuracy} \hspace{0.1 in}
\end{align}
\end{subequations}
Here, the moments for the first equation are with respect to the distribution $\Pi_{true}(d)$, and for the third equation, we abbreviated $\expect_{\Pi_{true}(d),\Pi_{nom}(d)}[\cdot] = \expect_{\Pi(d)}[\cdot]$.

Now we can use Theorem~\ref{thm:performance_bound} to construct a lower bound to $\rho^*$.  Doing so requires us to use Algorithm~\ref{alg:algorithm} to compute $i)$ a lower bound $\Tilde \rho$ to the minimum expected nominal robustness $\hat \rho^*$ and $ii)$ an upper bound $\Tilde \epsilon$ to the maximum expected signal difference $e$.  We approximate samples of each objective function, $\hat \rho$ and $e$, by calculating the robustness of one nominal trajectory and calculating the (semi)norm error between one nominal and true system trajectory, respectively.

\begin{table}[b]
\caption{Algorithm~\ref{alg:algorithm} parameters/results by optimization problem}
\centering
\begin{tabular}{|c|c|c|c|c|c|c|}
\hline
& $B$ & $R$ & $\delta$ & $\alpha$ & $c$ & bounds\\ \hline
\eqref{eq:example_true_robustness} & 0.1 & 0.15 & 0.05 & 0.05 & 0.3 & $\rho^* \geq 0.31$ \\ \hline
\eqref{eq:example_sim_robustness} & 0.2 & 0.1 & 0.05 & 0.05 & 0.2 & $\Tilde \rho = 0.46$ \\ \hline
\eqref{eq:example_accuracy} & 0.1 & 0.05 & 0.05 & 0.01 & 0.1 & $\Tilde e = 0.38$ \\ \hline 
\end{tabular}
\label{table:params}
\end{table}

\begin{figure}[t]
    \centering
    \includegraphics[width = 0.48\textwidth]{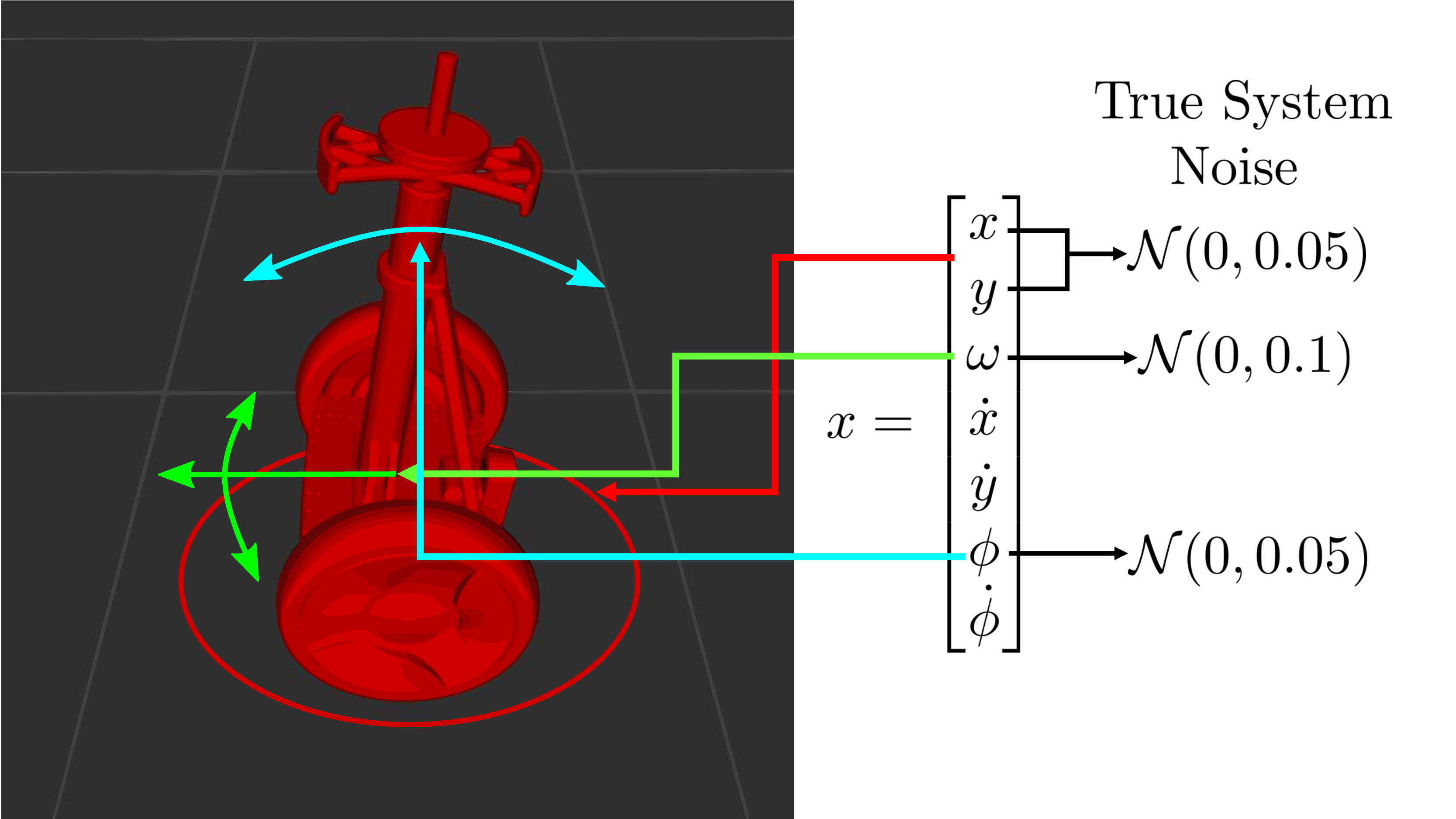}\\
    \vspace{-0.1 in}
    \caption{Example of the simulation environment used in Section~\ref{sec:bounding_robustness}.  The true system is always initialized to the center of one cell in a $5\times 5$ grid, facing directly left - heading, $\omega = 0$ - and with its pendulum upright - angle, $\phi = 0$.  The rightmost column indicates the additive noise with which the true system is initialized, \textit{i.e.} its initial condition on the plane is corrupted by zero-mean Gaussian Noise, $\xi \sim \mathcal{N}(0,0.05)$.}
    \label{fig:setup}
    \vspace{-0.2 in}
\end{figure}

Table~\ref{table:params} shows the parameters and results of applying Algorithm~\ref{alg:algorithm} to solve optimization problems~\eqref{eq:example_true_robustness}-\eqref{eq:example_accuracy}. As per Theorem~\ref{thm:performance_bound}, the results of Table~\ref{table:params} indicate a lower bound $\ell = \Tilde \rho - \Tilde e - r(M+m)/2 =0.46 - 0.38 - 0.08 = 0$ implying
\begin{align}
    \label{eq:calculated_bound}
    \hspace{-0.3 in}\prob_{\pi_1,\pi_2}[\rho^* \geq 0] & \geq \Delta(0.2,0.05,0.1)\Delta(0.1,0.05,0.05) \\
    & \geq 0.84.
\end{align}
We can also compare the bound above with the lower bound generated via direct application of Algorithm~\ref{alg:algorithm} to solve optimization problem~\eqref{eq:example_true_robustness}.
\begin{equation}
    \label{eq:true_bound}
    \prob_{\pi}[\rho^* \geq 0.31] \geq \Delta(0.3,0.05,0.15) \geq 0.92
\end{equation}
As the calculated lower bound in equation~\eqref{eq:true_bound}, $0.31$, is greater than the true lower bound in equation~\eqref{eq:calculated_bound}, $0$, and the calculated probability bound for equation~\eqref{eq:calculated_bound}, $0.84$, is less than the true probability bound for equation~\eqref{eq:true_bound}, $0.92$, this shows that the lower bound generated via our procedure is indeed expressive of true system phenomena.

\begin{figure}[t]
    \centering
    \includegraphics[width = 0.48\textwidth]{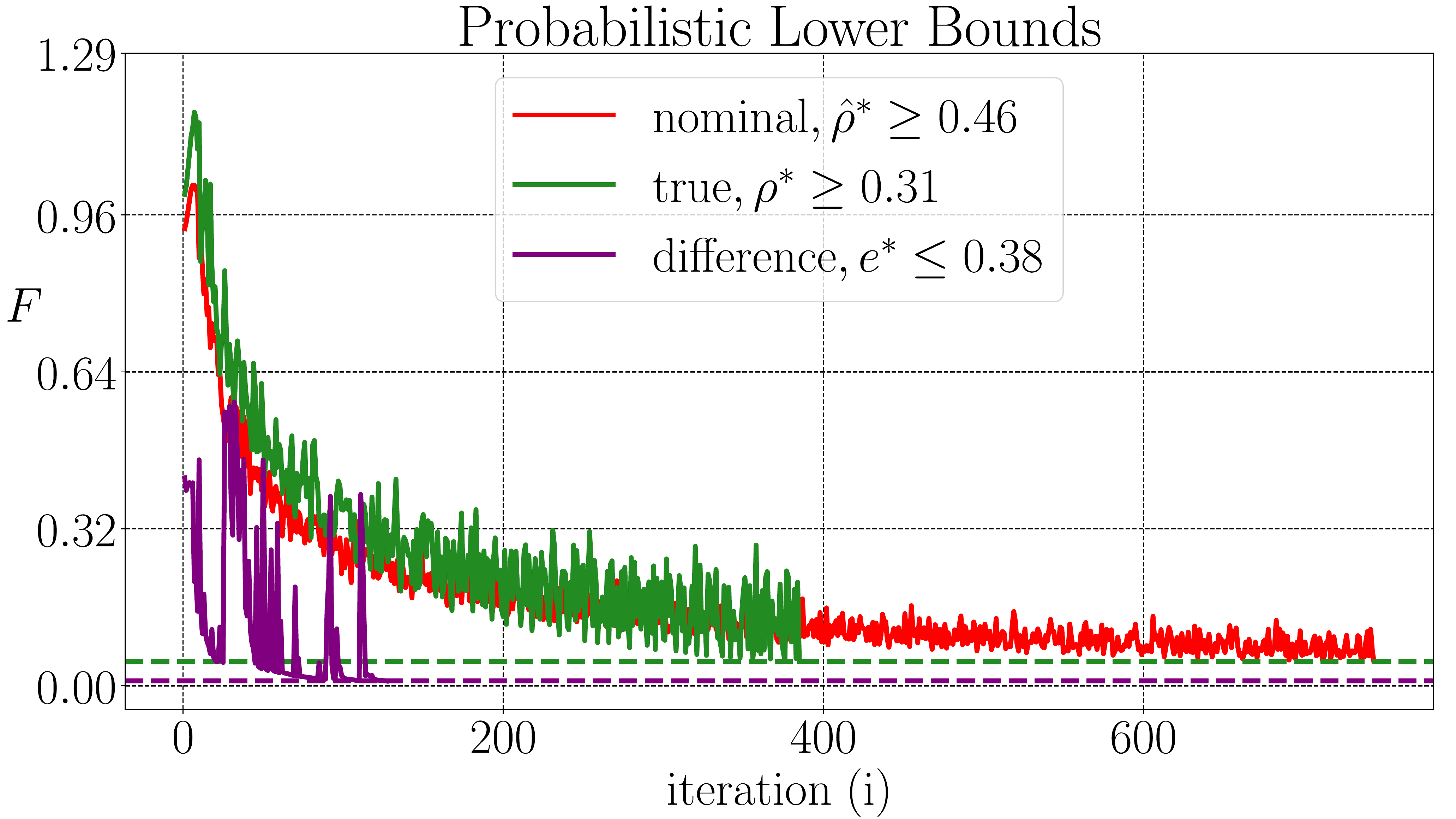}\\
    \gapclose
    \caption{The results of applying Algorithm~\ref{alg:algorithm} to bound the solutions to optimization problems~\eqref{eq:example_true_robustness}-\eqref{eq:example_accuracy}.  The dashed lines of similar colors indicate the tolerance for each algorithm run.  As the tolerances for the runs to solve optimization problems~\eqref{eq:example_true_robustness} and \eqref{eq:example_sim_robustness} were the same, we only see one, dashed line in green.}
    \vspace{-0.2 in}
    \label{fig:experimental_decay}
\end{figure}

As motivated prior though, not only did we want to lower bound $\rho^*$, but we also wished to do so \textit{while minimizing the number of iterations (system tests) required to determine this bound}.  To show that our method achieves this, Figure~\ref{fig:experimental_decay} shows the decay of $F$ while running Algorithm~\ref{alg:algorithm} to lower bound the solutions to optimization problems~\eqref{eq:example_true_robustness}--\eqref{eq:example_accuracy}.  Notice that Algorithm~\ref{alg:algorithm} \textit{required $i^* = 385$ tests to directly solve optimization problem~\eqref{eq:example_true_robustness}, but via our proposed method, we could lower bound $\rho^*$ in $i^* = 128$ tests by solving optimization problems~\eqref{eq:example_sim_robustness} and~\eqref{eq:example_accuracy} instead}.  As a result, our proposed method \textit{required $257$ fewer tests of the true system to construct a lower bound to the true-system risk measure $\rho^*$ as compared to a direct Bayesian Optimization testing scheme~\cite{ghosh2018verifying,gangopadhyay2019identification}.}

\section{Conclusion and Future Work}
The authors first developed a Bayesian Optimization Algorithm that identifies minimal upper bounds $\epsilon$ to the maximum of a function satisfying a set of assumptions.  Then, the authors used this Algorithm to construct a lower bound for a true system, robustness risk measure by solving two, separate optimization problems over the system simulator.  Finally, we showed that this procedure generates a lower bound that is emblematic of true system behavior, while minimizing the number of true system tests required to achieve that bound.  However, the generated bound is conservative and is restricted to a specific class of STL specifications.  In future work, the authors hope to extend the class of STL specifications accountable via our procedure and decrease the conservativeness of the resulting bound as well.
   
\renewcommand{\baselinestretch}{0.97}

\bibliographystyle{ieeetr}
\bibliography{bib_works}

\end{document}